%% file: kmodules_arxivL4.tex
\documentclass[a4paper,10pt]{article}
\usepackage{etex}
\usepackage{amssymb}
\usepackage{amsmath}
\usepackage{algorithm}
\usepackage{algpseudocode}
\usepackage{rotating}
\usepackage{enumitem}
\usepackage{mathbbol}
\usepackage{float}
\usepackage{fancyhdr}
\usepackage{geometry}
\usepackage[standard, thmmarks]{ntheorem}
\usepackage[colorlinks=false, pdfborder={0 0 0}]{hyperref}
\usepackage{xcolor}
\usepackage{color,tikz}

\newtheorem{observation}{Observation}

\newcommand{\A}{K}
\newcommand{\B}{L}
\newcommand{\C}{H}
\newcommand{\X}{S}
\newcommand{\Rxn}{\mathcal{N}}

\newcommand{\Met}{\mathcal{R}}

\newcommand{\R}{\mathbb{R}}
\newcommand{\rank}{\mathrm{rank}}

\newcommand{\conv}{\mathrm{conv}}
\newcommand{\pr}{\mathrm{pr}}
\newcommand{\supp}{\mathrm{supp}}

\newcommand{\Mod}{\mathfrak{B}}

\newcommand{\cNP}{\ensuremath{\mathbf{NP}}}

\newcommand{\eps}{\varepsilon}

\newcommand\blfootnote[1]{%
  \begingroup
  \renewcommand\thefootnote{}\footnote{#1}%
  \addtocounter{footnote}{-1}%
  \endgroup
}

\makeatletter
\newcommand{\problemtag}[1]{%
  (\textsc{#1})%
  \def\@currentlabel{\textsc{#1}}%
  \phantomsection%
}
\makeatother

\newcommand{\citep}{\cite}

\title{Polynomial time vertex enumeration of convex polytopes of bounded branch-width}
\author{Arne C. Reimers$^{1,*}$, Leen Stougie$^{1,2,3}$}

\begin{document}

\parindent 0em
\parskip 0.5em

\maketitle

\blfootnote{\hspace{-0.53cm}
%$^1$ Google TODO(arne): get clearance We really should check what the journal policy is - wether it is about where the work was conducted. It is much easier if Google is not one of my affiliations. \\
$^1$ CWI, Science Park 123, 1098 XG Amsterdam, The Netherlands \\
$^2$ Operations Research, VU University, De Boelelaan 1085, 1081 HV, Amsterdam, The Netherlands \\
$^3$ ERABLE team, INRIA Grenoble Rh\^one-Alpes, France \\
$^*$ Arne Reimers works now for Google Germany}

\begin{abstract}
\noindent Over the last years the vertex enumeration problem of polyhedra has seen a revival in the study of metabolic networks, which increased the demand for efficient vertex enumeration algorithms for high-dimensional polyhedra given by inequalities.
It is a famous and long standing open question in polyhedral theory and computational geometry whether the vertices of a polytope (bounded polyhedron), described by a set of linear constraints, can be enumerated in total polynomial time.
%~\cite{avis1992pivoting,khachiyan08}. 
%i.e. in a running time bounded by a polynomial function of input and output size.
In this paper we apply the concept of branch-decomposition to the vertex enumeration problem of polyhedra $P = \{x : Ax =  b, x \geq 0\}$.
For this purpose, we introduce the concept of $k$-module and show how it relates to the separators of the linear matroid generated by the columns of $A$. We then use this to present a total polynomial time algorithm for polytopes $P$ for which the branch-width of the linear matroid generated by $A$ is bounded by a constant $k$.

%\TODO{what do we do with the biological applications? Do we drop them entirely?}
%No we should mention them. Specifically we should mention that vertex enumeration of polytopes have seen a revival because of the study of metabolic networks,

%\TODO{What do we do with the other decomposition theorems for $1$- and $2$-modules?}

% For polyhedra of optimal flows in metabolic networks we found that in practice the vertices can be described as a cartesion product of the vertices of the flow-polyhedra of sub-networks (called flux modules). Using matroid theory, we showed that a decomposition into minimal sub-networks can be found in polynomial time.
% Since every polyhedron can be described as the flow-polyhedron of a metabolic network, the results theoretically are not limited to metabolic networks. However, the structure that allows such a decomposition is rather special.
%
% In this paper we introduce the concept of $k$-module, which generalizes the notion of flux module. We show how $k$-modules can be found efficiently and how they can be used to describe the vertices of the polyhedron.
% This opens the door for new vertex enumeration methods.
\end{abstract}

\bigskip


\input{introductionL4}

\input{basicsL4}

\input{decompositionL4}

\input{conclusionL4}

\section{Acknowledgements}
The work of Arne Reimers was supported by a PhD-scholarship of the Berlin Mathematical School at Freie Universit\"at Berlin and has been carried out during the tenure of an ERCIM Alain BensoussanFellowship Programme at Centrum Wiskunde \& Informatica, Amsterdam. The work of Leen Stougie was partially supported by the Einstein Foundation while visiting the Technische Universit\"at Berlin.

\bibliographystyle{plain}
\bibliography{modules}

\input{appendix}

\end{document}

%% file: introductionL4.tex
\section{Introduction}

A famous and long standing open question in polyhedral theory and computational geometry is whether the vertices of a polytope (bounded polyhedron), described by a set of linear constraints, can be enumerated in total polynomial time Avis and Fukuda~\cite{avis1992pivoting}, Boros et al.~\cite{boros09}, Khachiyan et al.~\cite{khachiyan08}, i.e. in a running time bounded by a polynomial function of input and output size.
For non-degenerate polyhedra pivoting based algorithms exist that run in \emph{total polynomial time} Dyer~\cite{dyer83}, Avis and Fukuda~\cite{avis1992pivoting}. Moreover, it is not possible to enumerate the vertices of general (unbounded) polyhedra in total polynomial time unless $\mathbf{P}=\mathbf{NP}$ Khachiyan et al.~\cite{khachiyan08}. But the existence of a total polynomial time vertex enumeration algorithm for polytopes remains unresolved.

A revival in significance of this question has come along with the advent of genome scale metabolic networks in the field of computational systems biology, see e.g. Schuster and Hilgetag~\cite{schuster94}, Gagneur and Klamt~\cite{gagneur04}, Acu\~na et al.~\cite{acuna12}. Vertex enumeration and variants thereof became an important biological analysis tool and several double description, Motzkin et al.~\cite{motzkin53}, based methods have been designed for the purpose, Terzer and Stelling~\cite{terzer2008large}, Hunt et al.~\cite{hunt14}.

Our theoretical approach of the vertex enumeration problem is of a fundamentally different nature than used in previous types of methods. We take a {\em branch decomposition} based perspective, very much in the spirit of research on the more famous notion of tree-width of graphs, Halin~\cite{halin76}, Robertson and Seymour~\cite{robertson84}, Hlin$\check{\rm e}$n\'y and Whittle~\cite{hlineny06}, which often leads to polynomial time solvability of graph optimization problems that are NP-hard in general.
{\em Branch-width}, defined in Robertson and Seymour~\cite{robertson91}, Hicks and Oum~\cite{hicks11}, is applicable to {\em matroids} Hicks and Oum~\cite{hicks11}.
On its turn, for polyhedra of the form $P = \{x\in \R^n : Ax = b, x \geq 0\}$,
we define the branch-width in terms of the linear matroid defined
by the linear independence relationship between subsets of the columns of the coefficient matrix
$A$ of $P$. The branch decomposition of the set of columns of $A$ leads to the notion of $k$-module for subsets of columns of $A$, which corresponds directly to $k$-separation in matroid theory Oxley~\cite{oxley11}.

The notion of $k$-module extends the concept of {\em flux module}, previously introduced for specific enumeration problems on metabolic networks in M\"uller and Bockmayr~\cite{mueller13}, Reimers et al.~\cite{reimers15}.
It's usefulness becomes clear when we will consider vertices of a polyhedron not so much as vectors in $\mathbb{R}^n$ but as combinatorial objects, the set of basis columns that define them (see e.g. Bertsimas and Tsitsiklis~\cite{bertsimas97}).

%to $k$-modules for the decomposition of arbitrary polyhedra and even non-convex subsets of $\R^n$.

The notion of branch-width we will employ has already been used by Cunningham and Geelen \cite{cunningham07} to show tractability of a class of integer programs. We show here its usefulness as a complexity measure for polytopes. In particular, we present a total polynomial time algorithm for vertex enumeration of polytopes if the branch-width of the underlying linear matroid is bounded by a constant $k$.

To the best of our knowledge this is the first significant result in the quest for the complexity of vertex enumeration of polytopes since Khachiyan et al. \cite{khachiyan08} proved its hardness for any (unbounded) polyhedra in 2008. Yet, it leaves the question open if vertex enumeration of polytopes can be done in total polynomial time.

The paper is organized as follows. In Section~\ref{subsec:terminology} we introduce the somewhat unconventional terminology which allows us to define faces and vertices of polytopes as combinatorial objects. In the same section we give the definition of branch decomposition and branch-width for matroids. In Section~\ref{subsec:k-modules} we introduce $k$-modules and characterize them through their correspondence to $k$-separation in matroid theory.
These ingredients are then used in Section~\ref{sec:totalP} to derive the central result of this paper: a total polynomial time vertex enumeration algorithm for polytopes of which the coefficient matrix has branch-width bounded by a constant. A discussion of our results and some open polyhedral questions form the final Section~\ref{sec:conclusion}.

%% file: basicsL4.tex
\section{Preliminaries}
\label{subsec:terminology}
In this paper, we will concentrate on polyhedra of the form:
\begin{eqnarray}
\label{eq:defP}
P = \{x\in \R^\Rxn : Ax = b, x \geq 0\},
\end{eqnarray}
the so-called standard LP-formulation.
Let us first notice the somewhat unconventional notation. As will become clear, it is very useful to have named variables (i.e., named columns of $A$). We therefore do not consider the elements of $P$ as elements of $\R^n$, but as elements of $\R^\Rxn$, which is the set of functions $x : \Rxn \to \R$, where $\Rxn$ is the set of variable names. Our indexing notation is thus defined as $x_i := x(i)$. Furthermore, we index with subsets $\A \subseteq \Rxn$ and define $x_\A : \A \to \R$ with $x_\A(i) = x(i)$ for $i \in \A$.
This allows chained indexing, which does not work in general with classical position-based indexing: i.e., for $\A \subseteq \B \subseteq \Rxn$, $x \in \R^\Rxn$, $y = x_\B$, and $z = y_\A$ implies $z = x_\A$.

Similarly, the matrix $A$ is an element of $\R^{\Met \times \Rxn}$, where $\Met$ denotes the set of rows ($\Rxn$ is the set of columns). We use the same indexing notation as for the vectors to write
$A_\A$ for the sub-matrix consisting of the columns $\A$.
For our computational results we will always assume that inputs are rational numbers.

We notice here that the focus on the particular form of the polyhedra is not a real restriction, since it can always be obtained by adding slack variables and redefining the $x$ variables, as argued in Appendix \ref{appendix:formulation_polyhedron}.

The standard definition of a face of polyhedron $P$ is the non-empty intersection of $P$ with a hyperplane that has the property that $P$ lies entirely on one side of it (see e.g. Schrijver~\cite{Schrijver86}). We define here faces of $P$ in general and vertices of $P$ in particular as discrete objects by means of their support.
\begin{definition}[Support of a face]
 Let $Q \subseteq P$ be a face of $P$.
  The support of $Q$ is the set $\supp(Q) := \{i \in \Rxn : \exists x \in Q\;\ x_i > 0 \}.$
\end{definition}
Since any face $Q$ of $P$ must satisfy $Ax=b$, it is uniquely defined by the set of variables for which the non-negativity constraints are tight, and hence, alternatively, by its complement, i.e. the set of variables in the support of $Q$. This leads us to the combinatorial definition of a face.
\begin{definition}[Face]\label{def:face}
We call a set $F \subseteq \Rxn$ a face of $P$ if there exists a face $Q \subseteq P$ of $P$ with $\supp(Q) = F$.
\end{definition}
It leads to the following characterization:
\begin{observation}\label{obs:face}
 $F \subseteq \Rxn$ is a face of $P$ if and only if there exists a $x \in P$ with $\supp(x) = F$.
\end{observation}
This allows us to call a face $F \subseteq \Rxn$ with minimal support a vertex of $P$.
\begin{proposition} \label{prop:vertex_definition}
 $x \in P$ is a vertex of $P$ if and only if $F := \supp(x)$ is minimal.
\begin{proof}
 Immediate from the fact that $x \geq 0$ are the only inequality constraints of $P$.
\end{proof}
\end{proposition}
Our algorithm for enumerating the vertices of $P$ is based on a branch decomposition of the linear matroid represented by the columns of the coefficient matrix $A$. Branch decomposition and branch-width have been defined first in Robertson and Seymour~\cite{robertson91}. For independent reading we define these notions here for matroids. We assume that the definitions of (linear) matroid and rank are known (otherwise we refer the reader to e.g. Oxley~\cite{oxley11}), but we do define here the notion of $k$-separation, which is crucial in this paper.

Let $M$ be a matroid on a set $\Rxn$ of elements. For $\A\subset \Rxn$ let $\lambda(\A) := \rank(\A) + \rank(\Rxn \setminus \A) - \rank(\Rxn)$ be the so-called connectivity function of $M$.
\begin{definition}[$k$-separation] \label{def:k-separator}
$\A\subset \Rxn$ is a $k$-separation if $\lambda(\A) < k$.
\end{definition}

\begin{definition}[branch decomposition and branch-width] \hspace{1cm}
\begin{itemize}
\item A \emph{branch decomposition} of $M$ consists of a tree $T$ with nodes of degree $3$ and $1$ only. Its $|\Rxn|$ leaves are identified with the elements $\Rxn$ of $M$.
\item The \emph{width} of an edge $e$ of $T$ is $\lambda(\A_e)$, where $\A_e$ is the set of leaves leaves of $T$ (elements of $\Rxn$) of any of the two components obtained by deleting the edge $e$ from $T$. The width is well defined, since $\lambda(\A) = \lambda(\Rxn \setminus \A)$ for any $\A \subset \Rxn$.
%$(\A_e,\B_e)$ is the partition of the leaves of $T$ (elements of $\Rxn$) given by $T \setminus e$. Observe that deleting an edge of $T$ splits $T$ into two connected components, one with leaves $\A_e$ and the other with leaves $\B_e$. The width is well defined, since $\lambda(\A) = \lambda(\Rxn \setminus \A)$.
\item The width of a branch decomposition is the maximum width of an edge $e \in T$.
\item The \emph{branch-width} of $M$ is the minimum width of all possible branch de\-com\-po\-si\-tions.
\end{itemize}
\end{definition}

Hence, the width of an edge $e$ of $T$ is $k$ if and only if $\A_e$ is a $k+1$-separation of the matroid $M$.

Using methods described by Bixby and Cunningham~\cite{bixby79,bixby90} we can find $k$-separations in polynomial time if $k$ is assumed fixed.
Algorithms that directly compute decompositions of matroids into separations have been studied in the context of branch-decompositions by Oum and Seymour~\cite{oum06,oum07}.

We conclude this section with some frequently used notation. The set of all vertices of $P$ is denoted by $\mathcal{V}$.
We use $\dot\cup$ to denote disjoint union,
$\subset$ to denote a proper subset,
and $\langle \cdot \rangle$ to denote the linear hull.

\section{k-Modules} \label{subsec:k-modules}
To understand what branch-width and $k$-separation mean in the context of polyhedra, we introduce the notion of $k$-modules and show that it is equivalent to $(k+1)$-separation if all variables exhibit variability (Theorem ~\ref{thm:lin_space_gen} and \ref{thm:separator_gen}). As we will argue soon, this is no restriction.

In this section we will use $\bar{P}$ for a subset of a polyhedron, because many of the following results are not restricted to polyhedra of the form of (\ref{eq:defP}).
\begin{definition}[$k$-module] \label{def:aff_k_mod}
Let $\bar{P} \subseteq \{x \in \R^\Rxn : Ax = b \}$.
$\A \subseteq \Rxn$ is a {\em $\bar{P}$-$k$-module} if there exists a $d \in \R^\Met$ and a $D \in \R^{\Met \times k}$ s.t.
\begin{align*}
\forall x \in \bar{P} \ \exists \alpha \in \R^k : A_\A x_\A = d + D\alpha.
\end{align*}
We call $d$ the {\em constant interface} vector of the module and $D$ the {\em variable interface} matrix.
If we can choose $d=0$ we will say that the $\bar{P}$-$k$-module is a {\em linear} $\bar{P}$-$k$-module.
\end{definition}
If the polyhedron of reference is clear from the context, we will simply write $k$-module instead of $\bar{P}$-$k$-module. We notice that the flux modules, defined in M\"uller and Bockmayr~\cite{mueller13} and mentioned in the introduction, are $0$-modules in the context of this definition, only having constant interface.

Without proof, we give some observations, which may help the reader to get some intuition for the notion.
\begin{observation} \label{obs:k-module}
Let $\bar P \subseteq \{x \in \R^\Rxn : Ax = b \}$ and $\A\subset \Rxn$.
\begin{itemize}
\item [(i)] $\A$ is a $k$-module if and only if $\dim \{A_\A x_\A : x \in \bar P\} \leq k$.
\item [(ii)] Every set $\A$ with $k$ elements is a (linear) $k$-module; in particular every $i \in \Rxn$ is a (linear) $1$-module;
\item [(iii)] $\A$ is a $(k-1)$-module $\Rightarrow$ $\A$ is a $k$-module;
\item [(iv)] Let $\B \subseteq \Rxn$ be a $0$-module. It holds for all $\A \subseteq \Rxn \setminus \B$ that $\A$ is a $k$-module if and only if $\A \dot\cup \B$ is a $k$-module;
\item [(v)] Let $\A$ be a set of variables without variability; i.e., have fixed values in $\bar P$. Then $\A$ is a $0$-module.
\end{itemize}
\end{observation}
The following proposition shows that for a given $k$-module $\A$ the variable interface $D$ is not unique but its span $\langle D \rangle := \{D \alpha : \alpha \in \R^k\}$ is, if $k$ is chosen to be minimal, in the context of (iii) above.
\begin{proposition}
 Let $\bar{P} \subseteq \R^\Rxn, \A \subseteq \Rxn$ and let $k$ be minimal s.t. $\A$ is a $k$-module of $\bar{P}$. Let $D, D'$ be two different variable interfaces of $\A$. Then $\langle D \rangle = \langle D' \rangle$ and the dimension of $\langle D \rangle$ is equal to $k$.
\begin{proof}
 Assume $\langle D \rangle \neq \langle D' \rangle$. By definition of variable interface, it follows that
\begin{eqnarray}
\label{eq:spnintersection}
\left.
\begin{array}{r}
\forall x \in \bar{P} : A_\A x_\A \in d+ \langle D \rangle \\
\forall x \in \bar{P} : A_\A x_\A \in d+ \langle D' \rangle \\
\end{array}
\right\}
\Rightarrow \forall x \in \bar{P} &: A_\A x_\A \in d+ \langle D \rangle \cap \langle D' \rangle.
\end{eqnarray}
Since $\langle D \rangle \neq \langle D' \rangle$, it follows that $\langle D \rangle \cap \langle D' \rangle \subset \langle D \rangle$. Hence,
$$
\dim(\langle D \rangle \cap \langle D' \rangle) < \dim(\langle D \rangle).
$$
It follows that there exists a $D'' \in \R^{\Met \times \ell}$ with $\ell < k$ and $\langle D'' \rangle = \langle D \rangle \cap \langle D' \rangle$. By \eqref{eq:spnintersection} it follows that $D''$ is a variable interface of $A$ and hence, $k$ was not minimal; a contradiction.
\end{proof}
\end{proposition}

% \BCA Observation (iii) implies in particular that every $r\notin V$, a $0$-module, is a $k$-module. And Observation (v) implies that if $A$ is a $k$-module, then $A\cap V$ is an affine $k$-module. \ECA

Extending the case of flux modules in Reimers et al.~\cite{reimers15}, we will prove that for general $k$-modules we can restrict ourselves to the analysis of linear vector spaces. To facilitate the exposition, we will assume from here on that no $i\in \Rxn$ exists such that $x_i$ has fixed value in $\bar{P}$. We notice that this is not a severe restriction, which can always be met for convex polyhedral sets $\bar{P}$ by preprocessing using LP and replacing variables by their fixed values, whenever applicable (see Observation~\ref{obs:k-module}).

We use $\ker(A)$ to denote the kernel of $A$: $\ker(A)= \{x \in \R^\Rxn : Ax = 0\}$
\begin{theorem} \label{thm:lin_space_gen}
 Let $\bar{P} \subseteq \{x \in \R^\Rxn : Ax = b\}$. Then for all $\A \subseteq \Rxn$ we have
\begin{align*}
\A \mbox{ is a $k$-module of $\bar{P}$} &\Leftrightarrow \A \mbox{ is a $k$-module of $\ker(A)$}.
\end{align*}
\end{theorem}
Before proving the theorem, we observe that for any set that contains $0$, hence any linear vector space and in particular $\ker(A)$, this theorem implies that we can concentrate on linear $k$-modules: for
$0\in \ker(A)$ implies that for any $k$-module of $\ker(A)$ we can choose constant interface $d=0$.
\begin{proof}
 \begin{enumerate}
  \item[$\Leftarrow$:] Let $x^1 \in \bar{P}$ be arbitrary but fixed. We define $d = A_\A x_\A^1$. Let $x^2 \in \bar{P}$ be arbitrary. By definition of $\bar{P}$, we have $x^1 - x^2 \in \ker(A)$. Since, as argued, $\A$ is a linear $k$-module of $\ker(A)$, there exists an $\alpha \in \R^k$ such that $A_\A (x^2_\A - x^1_\A) = D \alpha$, where $D$ is the variable interface of $\A$.
Thus, $A_\A x^2_\A = A_\A x^1_\A + D \alpha = d + D \alpha$ and therefore $\A$ is a $k$-module of $\bar{P}$.
  \item[$\Rightarrow$:]
Let $\A$ be a $k$-module of $\bar{P}$. Let $d \in \R^k, D \in R^{\Met \times k}$ be the constant and fixed interfaces of $\A$ in $\bar{P}$. Take any $y\in \bar{P}$ with the property that there exists $\epsilon > 0$ s.t. $y+w \in P$ for all $w \in \ker(A)$ with $\Vert w \Vert_\infty < \eps$. By our assumption that no coordinate of any $x\in \bar{P}$ has fixed value, such a $y$ must exist. By definition of $k$-module of $\bar{P}$, $\alpha^y \in \R^k$ exists s.t. $A_\A y_\A = d + D \alpha^y$.
Now assume that $\A$ is not a $k$-module of $\ker(A)$. Then there exists $w \in \ker(A)$ s.t. for all $\alpha \in \R^k$ we have $A_\A w_\A \neq D \alpha$. By our choice of of $y$, there exists an $\eps' > 0$ s.t. $y + \eps' w \in P$. We conclude that
\begin{align*}
 A_\A (y_\A + \eps' w_\A) &= A_\A y_\A + \eps' A_\A w_\A \neq d + D \alpha^y + D \alpha &\mbox{for all } \alpha \in \R^k. \\
\Rightarrow A_\A (y_\A + \eps' w_\A) &\neq d + D \alpha &\mbox{for all } \alpha \in \R^k.
\end{align*}
 This is a contradiction. \qed
 \end{enumerate}
\end{proof}
We notice that the theorem applies in particular to polyhedra of the form $P=\{x\in \R^\Rxn : Ax = b, x \geq 0\}$.

This result allows us to restrict our attention to $k$-modules of linear vector spaces of the form $\{ x\in \R^\Rxn : Ax=0 \}$, which, as the following theorem states, are tightly related to the notion of separation of the linear matroid defined on the columns of the matrix $A$.
\begin{theorem} \label{thm:separator_gen}
 $\A \subseteq \Rxn$ is a $k$-module of $\ker(A)$ if and only if $\A$ is a $k+1$-separation of the linear matroid $M$ defined by linear independence of the columns of $A$.
\end{theorem}
\begin{proof}
Let $\B$ be the complement of $\A$. For $\A$ a $k$-module of $\ker(A)$ there exists matrix $D$ with $k$ columns such that $A_\A x_\A\in \langle D \rangle$ for all $x\in \ker(A)$. For any $z\in \langle A_\A \rangle \cap \langle A_\B \rangle$ there exist $x'_\A$ and $x''_\B$ such that $A_\A x'=z=A_\B x''_\B$. Hence $x\in \R^\Rxn$ composed of $x_\A =x'_\A$ and $x_\B=-x''_\B$ has $Ax=0$, whence $x\in \ker(A)$.
%Therefore $z=S_Ax_A=S_Ax'_A \in \langle D \rangle$.
Therefore, $\langle A_\A \rangle \cap \langle A_\B \rangle \subseteq \langle D \rangle$. This implies that \begin{eqnarray*}
\lambda(\A) &=& \rank(\A) + \rank(\Rxn \setminus \A) - \rank(\Rxn) \\
            &=& \dim \left( \langle A_\A \rangle \right) + \dim \left(\langle A_\B \rangle \right) - \dim \left( \langle A_\A \rangle \cup \langle A_\B \rangle \right) \\
            &=& \dim \left( \langle A_\A \rangle \cap \langle A_\B \rangle \right)\leq k.
\end{eqnarray*}
For the other direction, suppose that $\lambda(\A)\leq k$. Since $\lambda(\A)=\dim(\langle A_\A \rangle \cap \langle A_\B \rangle)$, there exists a matrix $D$ with $k$ columns such that $\langle A_\A \rangle \cap \langle A_\B \rangle = \langle D \rangle$. For any $x\in \ker(A)$ we have $A_\A x_\A+A_\B x_\B=0$ and therefore $A_\A x_\A \in \langle A_\B \rangle$. Hence, $A_\A x_\A\in \langle A_\A \rangle \cap \langle A_\B \rangle$ and there exists $\alpha \in \R^k$ such that $A_\A x_\A=D \alpha$.
\end{proof}

%% file: decompositionL4.tex
\section{Vertex Enumeration} \label{sec:totalP}
The ingredients from the previous section will be used here to design a total polynomial time vertex enumeration algorithm for polytopes of bounded branch-width. In fact we notice that all results so far hold irrespective of boundedness of the polyhedra and this will continue to be true for most of this section as well.
We will indicate clearly where boundedness will be required.

Remember from the introduction that any face of $P = \{x\in \R^\Rxn : Ax = b, x \geq 0\}$, and hence any vertex of $P$ is fully characterized by its support. We will build these supports of vertices by hierarchically merging disjoint subsets of the columns of $A$, starting from the individual columns. The process is guided by a branch decomposition of the matroid $M$ with elements $\Rxn$, the columns of $A$.

For our purposes we turn the branch decomposition into a hierarchical rooted decomposition with the leaf nodes of the (directed) tree forming a bijection with the single elements (columns of $A$) of the matroid $M$.
To do so, let $T$ be a branch-decomposition of $M$ with width at most $k+1$.
Choose an arbitrary edge $e = (a,b)$ of $T$. Subdivide $e$ and make the newly created vertex $r$ the root of a directed binary tree $T'$, with children $a,b$. Direct all edges away from $r$.
For each node $a$ of $T'$
we define the set $\A(a) := \{i \in \Rxn: i \mbox{ is leaf under }a\}$.
Thus, $T'$ describes a binary rooted tree with in the root $\A(r)=\Rxn$.
By definition of branch-decomposition, we observe for each node $a$ of $T'$ that $\A(a)$ is a $k+1$-separation of $M$.

As a result of this rooting operation we obtain a rooted branch decomposition of width $k+1$ for polyhedron $\bar P$, which gives us a hierarchical family of $k$-modules $\Mod =: \{\A(a) : a \mbox{ node of }T'\}$  satisfying the following properties, which are immediate, but serve the exposition of our enumeration algorithm.
\begin{enumerate}[label=(P\arabic*)]
 \item \label{branch_cond1} For each $\A \in \Mod$, $\A \neq \Rxn$ there exists exactly one $\B \in \Mod$ with $\A \dot\cup \B \in \Mod$.
 \item \label{branch_cond2} For each $\C \in \Mod, |\C| \geq 2$ there exists exactly one pair $\A,\B \in \Mod$ with $\A \dot\cup \B = \C$.
\end{enumerate}

From here on we assume that we are given the hierarchical (rooted) decomposition and the resulting nested family $\Mod$ of $k$-modules. For each module $\A \in \Mod$ let $D(\A) \in \R^{\Met \times k}$ denote the variable interface and $d(\A) \in \R^\Met$ the constant interface and define
\begin{align}
     P(\A) := \{x \in \R^\A : A_\A x = D(\A) \alpha + d(\A), x \geq 0, \exists \alpha \in \R^k \}.
\end{align}
To get a feeling for what we are aiming at, we notice that if $\A$ would happen to be a $0$-module then the vertices of $P$ would be given simply by the cartesian product of the vertices of $P(\A)$ and the vertices $P(\Rxn\setminus \A)$ Reimers et al.~\cite{reimers15}.

For $k$-modules if $k\neq 0$ this does not hold.
Therefore we will build the set of vertices by combining faces of the polyhedra $P(\A)$ for subsets $\A$ following the hierarchical structure of $\Mod$. Recall the combinatorial Definition~\ref{def:face} of a face and the  subsequent Observation~\ref{obs:face}. For obvious reasons we are only interested here in the the faces of $P(\A)$ that are restrictions of faces of $P$ to $\A$, i.e. faces of $P(\A)$ that can be extended to faces of $P$.

Therefore we define the notion of $\A$-face as the restriction of a face of $P$ to $\A$.
\begin{definition}[$\A$-face]
 For $\A \subseteq \Rxn$ a set $F \subseteq \A$ is called a \emph{$\A$-face} if there exists a $x \in P$ with $x_F > 0$ and $x_{\A \setminus F} = 0$.
\end{definition}
We remark that it follows immediately from the definition of $k$-module (resp. $P(\A)$) that every $\A$-face is a face of $P(\A)$. We also notice that $P(\A)$ may have more faces, but they are uninteresting for us.

In fact, for vertex enumeration we should only be interested in those faces of $P(\A)$ that are restrictions to $\A$ of vertices of $P$. :
\begin{definition}[vertex induced $\A$-face]
 For $\A \subseteq \Rxn$ a set $F \subseteq \A$ is called \emph{vertex induced $\A$-face} if there exists a vertex $x$ of $P$ with $x_F > 0$ and $x_{\A \setminus F} = 0$.
\end{definition}
We note for $\A \subseteq \Rxn$ that, whereas testing if a subset $F$ defines a $\A$-face can be done easily by linear programming (as we will show in Proposition~\ref{prop:check_min_feas}), testing if $F \subseteq \A$ is a vertex induced $\A$-face is NP-hard, Fukuda et al.~\cite{fukuda97}. That it remains hard even for polytopes of bounded branch-width we show in Appendix \ref{appendix:findVertex}.
As a relaxation of vertex induced $\A$-face, we propose the notion of an injective $\A$-face, which can obviously be tested in polynomial time:
%\BC This definition is only for $A \in \Mod$, because we use $P^A$. \EC
\begin{definition}[injective $\A$-face]
  A $\A$-face $F$ is an \emph{injective $\A$-face} if $A_F$ is injective.
%i.e., there do not exist distinct $y,z\in P$ with $A_Fy_F=A_Fz_F$.
\end{definition}
Let us show that injective $\A$-faces are indeed a relaxation of vertex induced $\A$-faces.
%\begin{lemma} \label{lemma:injectivity}
% An $\A$-face $F$ is injective if and only if there exist no distinct $y,z \in \{x \in P(\A) : \supp(x) \subseteq F\}$ with $A_\A y = A_\A z$, i.e. $A_\A$ is injective on $\{x \in P(\A) : x_{\A \setminus F} = 0\}$.
%\begin{proof}
%We observe that if $A_\A$ were not injective on $\{x \in P(\A) : x_{\A \setminus F} = 0\}$, then $A_F$ is not injective. We therefore only show the other direction.
%
%Assume that $A_F$ is not injective. Hence, there exists a $y \neq 0$ with $A_F y_F = 0$ and $y_{\A \setminus F} = 0$. Since $F$ is a $\A$-face, there exists a $x \in P(\A)$ with $x_{\A \setminus F} = 0$ and $x_{F} > 0$.
%Hence, for $\eps > 0$ small enough, $x + \eps y \in P(\A)$ and
%we have $A_\A(x + \eps y) = A_\A x$, contradicting the injectivity of $A_\A$ on $\{x \in P(\A): x_{\A \setminus F} = 0\}$.
%\end{proof}
%\end{lemma}
\begin{proposition} \label{claim2}
  Every vertex induced $\A$-face $F$ for $\A \subseteq \Rxn$ is an injective $\A$-face.
\begin{proof}
Let $F$ be a vertex induced $\A$-face, i.e. there exists a vertex $v \in P$ with $\supp(v_\A) = F$.
It follows immediately from the definition that $F$ is a $\A$-face.

Since $v$ is a vertex it follows that $A_{\supp(v)} x = b$ has $x = v$ as its unique solution and therefore, $A_{\supp(v)}$ is injective. It follows that $A_F$ is also injective and therefore $F$ is an injective $\A$-face.
\end{proof}
\end{proposition}
The crucial property that makes this relaxation in fact eventually exact is that vertex induced $\Rxn$-faces and injective $\Rxn$-faces coincide, as will be a corollary of the following theorem.
\begin{theorem} \label{thm:min_feas_vertex_feas}
 Let $\A \in \Mod$ be a $0$-module. Then, $F$ is an injective $\A$-face if and only if it is a vertex induced $\A$-face.
%\BCA For $0$-modules we have that $\A$-faces and feasible $\A$-faces are the same. Should this be a separate observation?\ECA
\begin{proof}
Proposition~\ref{claim2} shows that every vertex induced $\A$-face is an injective $\A$-face.

Let $F$ be an injective $\A$-face. Since $\A$ is a $0$-module it holds for all $x \in P(\A)$ that $A_\A x_\A = d$. It follows by definition of injective $\A$-face that $x \in P(\A)$ with $x_{\A\setminus F} = 0$ is unique. Since $F$ is a $\A$-face, $x$ exists and satisfies $x_F > 0$.

We observe that by its definition $P$ is pointed. Since $F$ is a $\A$-face, we have that $\hat P := \{x \in P : x_{\A\setminus F} = 0\}$ is a non-empty pointed polyhedron and a face of $P$. Hence, $\hat P$ has at least one vertex, $y$ say. Clearly, $y$ is also a vertex of $P$. Since $y\in \hat P$ and $\hat P$ is a face of $P$, $y$ satisfies $y_\A \in P(\A)$ and $y_{\A\setminus F} = 0$. Thus, $y_\A = x_\A$ and hence $F$ is a vertex induced $\A$-face.
\end{proof}
\end{theorem}
\begin{corollary}
  The collection of all injective $\Rxn$-faces is the collection of all vertices of $P$.
\begin{proof}
  Noticing that a vertex induced $\Rxn$-face defines exactly a vertex of $P$ and every vertex of $P$ corresponds to a vertex induced $\Rxn$-face (see Proposition~\ref{prop:vertex_definition}). Then the corollary follows from Theorem~\ref{thm:min_feas_vertex_feas} since $\Rxn$ is a $0$-module.
\end{proof}
\end{corollary}

We will show in the following theorem that Algorithm~\ref{alg:decomposition2} enumerates all injective $\C$-faces for a given $k$-module $\C\in \Mod$, by recursively enumerating all injective $\A$- and $\B$-faces for the two $k$-modules $\A,\B\in \Mod$ that constitute $\C$; i.e., $\C=\A\dot\cup \B$.

\begin{algorithm}
 \begin{algorithmic}
\State function $\mathcal{F}$ = getInjectiveFaces($\C$)
\If{$|\C| = 1$}
  \State $\mathcal{F} := \emptyset$.
  \If{$\emptyset$ is a $\C$-face}
		\State $\mathcal{F} := \mathcal{F} \cup \{\emptyset\}$.
  \EndIf
  \If{$\C$ is an injective $\C$-face}
		\State $\mathcal{F} := \mathcal{F} \cup \{ \C \}$.
  \EndIf
\Else
\State Let $\A,\B \in \Mod$ with $\C = \A \dot\cup \B$.
\State $\mathcal{F}^\A := $ getInjectiveFaces($\A$)
\State $\mathcal{F}^\B := $ getInjectiveFaces($\B$)
\State $\mathcal{F} := \{F^\A \cup F^\B : F^\A \in \mathcal{F}^\A, F^\B \in \mathcal{F}^\B \}$.
\For{$F \in \mathcal{F}$}
	\If{$F$ is not a $\C$-face or $A_F$ is not injective}
		\State $\mathcal{F} := \mathcal{F} \setminus \{F\}$.
	\EndIf
\EndFor
\EndIf
 \end{algorithmic}
 \caption{Algorithm to compute all injective $\C$-faces for $\C \in \Mod$. For $\C = \Rxn$, this algorithm will output all vertices.}
 \label{alg:decomposition2}
\end{algorithm}

\begin{theorem} \label{thm:correctness_alg_decomp2}
 Algorithm~\ref{alg:decomposition2} computes all injective $\C$-faces for a given $\C \in \Mod$.
\begin{proof}%[Thm.~\ref{thm:correctness_alg_decomp2}]
 First notice that the only two possible faces of a $k$-module $\C$ with $|\C| = 1$ are $\emptyset$, and $\C$ itself.
 Any $k$-module $\C\in \Mod$ with $|\C|\geq 2$ is constituted by two disjoint $k$-modules $\A,\B \in \Mod$: $\C =\A \dot\cup \B$.
%We will prove that for any minimal feasible face $F^C$ of $C \in \Mod$, there exist minimal feasible faces $F^A$ of $A$ and $F^B$ of $B$, such that $F^C = F^A \cup F^B$.
Now notice that for every injective $\C$-face $F^\C$, we have that $F^\A := F^\C \cap \A$ and $F^\B := F^\C \cap \B$
are, respectively, a $\A$-face and a $\B$-face. Moreover, also $F^\A$ and $F^\B$ are injective, since $A_{F^\C}$ being injective implies that also $A_{F^\A}$ and $A_{F^\B}$ are injective. Clearly, $F^\C = F^\A \cup F^\B$.
Since Algorithm~\ref{alg:decomposition2} tests for every possible pair consisting of an injective $\A$-face and an injective $\B$-face if their union is injective for $\C$, this implies the theorem for any set $\C\in \Mod$ with $|\C|\geq 2$.
\end{proof}
\end{theorem}

We emphasize that until here all results hold equally for bounded and unbounded polyhedra. Now we will show that in case $P$ is a polytope (i.e. $P$ bounded) the existence of a set $\Mod$ of $k$-modules makes Algorithm~\ref{alg:decomposition2} run in polynomial time, for fixed $k$ (bounded branch-width). Some observations will allow us to bound the number of injective faces and eventually obtain the runtime bound. The following proposition still also holds for unbounded polyhedra:
\begin{proposition}\label{prop:dimension_min_face}
 For every injective $\A$-face $F$, $\A \in \Mod$, we have that the corresponding face of $P(\A)$, i.e. $Q = \{x \in P(\A) : x_{\A\setminus F} = 0\}$, has dimension $\dim Q \leq k$.
\begin{proof}
Since $\A$ is a $k$-module, it follows that $A_\A$ maps every point in $Q$ into a $k$-dimensional space. If $\dim Q > k$, it would follow that $A_\A$ is not injective on $Q$.
\end{proof}
\end{proposition}
The following corollary may later give some further intuition for the final complexity bound.
%\BCN We never use the following Corollary. I suggest we remove it. \ECN
%\BCL We can also keep it even if we don't use it. We may even say it that we don't use it \ECL
\begin{corollary} \label{cor:claim2}
 Every vertex induced $\A$-face $F$, $\A \in \Mod$, satisfies $\dim \{x \in P(\A) : x_{A \setminus F} = 0\} \leq k$.
\begin{proof}
 Directly from Propositions~\ref{claim2} and \ref{prop:dimension_min_face}. %\BCN I added this one-line proof, because the propositions are not consecutive anymore. \ECN
\end{proof}
\end{corollary}

The following is a kind of Carath\'{e}odory's Theorem, see e.g. Schrijver~\cite{Schrijver86}, that will allow us to bound the number of injective $\A$-faces in terms of the number of vertex induced $\A$-faces. Here we need for the first time boundedness of $P$. Also we need some extra notation: for a set $X$ we write $|X|$ for its cardinality and $\conv(X)$ for its convex hull. By $\pr_\A $ we denote the projection of vectors in $P$ to their coordinates belonging to $\A$.
\begin{lemma} \label{claim4a}
Assume $P$ is bounded.
  Suppose for $\A \in \Mod$ that $F$ is a $\A$-face and let $h = \dim Q$ with
  $Q= \{x \in P(\A) : x_{\A\setminus F} = 0\}$.
  Then there exist a set of $\ell \leq h+1$ vertex induced $\A$-faces $F^1, \ldots, F^{\ell}$ such that $F = F^1 \cup F^2 \cup \ldots \cup F^{\ell}$.
\begin{proof}
 Since $F$ is a $\A$-face, there exists a $y \in P$ with $y_{\A\setminus F} = 0$ and $y_F > 0$. Therefore, $y$ lies on the face $\{x \in P: x_{\A\setminus F} = 0\}$ of $P$. Since $P$ is bounded, there exists a subset $V'\subset \mathcal{V}$ of the vertices of $P$ such that $y \in \conv(V')$ and $w_{\A\setminus F} = 0$ for all $w \in V'$.

 It follows that $y_\A \in \conv(\pr_\A V')$. Obviously, $\pr_\A V'\subseteq P(\A)$. Since $\dim \{x \in P(\A) : x_{\A\setminus F} = 0\} = h$,
there exist, by Carath\'{e}odory's theorem \cite{Schrijver86}, $\ell \leq h+1$ points $w^1_\A, \ldots, w^\ell_\A \in \pr_\A V'$, the projection of $\ell$ vertices of $P$ in $V'$, with $y_\A \in \conv(w^1_\A, \ldots, w^\ell_\A)$.

Clearly, $F^i = \{j \in \A : j\in \supp(w^i) \}$ is a vertex induced $\A$-face and $F \supseteq F^i$ for all $i=1,\ldots, \ell$.
Reversely, for every $j \in \A$ with $j \in \A \setminus F^i$ for all $i=1, \ldots, \ell$ it follows that $y_j = 0$, since $y_\A \in \conv(w^1_\A, \ldots, w^\ell_\A)$, and hence, $F \subseteq F^1 \cup F^2 \cup \ldots \cup F^\ell$. Thus, $F = F^1 \cup F^2 \cup \ldots \cup F^{\ell}$, which proves the lemma.
\end{proof}
\end{lemma}
\begin{proposition} \label{claim4}
If $P$ is bounded then for all $\A \in \Mod$
\begin{align*}
 \vert \{F \subseteq \A : F \mbox{ injective $\A$-face}\} \vert \leq  \vert \{F \subseteq \A : F \mbox{ vertex induced $\A$-face} \} \vert^{k+1}.
\end{align*}
\begin{proof}
By Proposition~\ref{prop:dimension_min_face} every injective $\A$-face $F$ gives rise to a face of $P(\A)$ of dimension at most $k$.
Hence, by Lemma~\ref{claim4a}, the union of a set of at most $k+1$ vertex induced $\A$-faces defines a unique $\A$-face $F$. In particular this is true for every injective $\A$-face. Denoting by $c_\mathrm{vert}$ the number of vertex induced $\A$-faces, there are at most
\begin{align*}
 \sum_{i=1}^{k+1} \binom{c_\mathrm{vert}}{i}
\end{align*}
non-empty subsets of at most $k+1$ elements.
%\BC But we should count all subsets of {\em at most} $k+1$ elements \EC
For $c_\mathrm{vert} = 1$, we have $\sum_{i=1}^{k+1} \binom{c_\mathrm{vert}}{i} = 1 = c_\mathrm{vert}^{k+1}$ and for $c_\mathrm{vert} \geq 2$, we can derive an upper bound
\begin{align*}
 \sum_{i=1}^{k+1} \binom{c_\mathrm{vert}}{i} \leq \sum_{i=1}^{k+1} \frac{c_\mathrm{vert}^i}{i!} \leq \sum_{i=1}^{k+1} \frac{2^{k+1-i} c_\mathrm{vert}^i}{k+1} \leq (k+1) \frac{c_\mathrm{vert}^{k+1}}{k+1} = c_\mathrm{vert}^{k+1},
\end{align*}
since $2^{k+1-i} \geq \frac{k+1}{i!}$ for all $i \leq k \in \mathbb{N}$.

By injectivity it follows that this is also a bound on the number of injective $\A$-faces.
\end{proof}
\end{proposition}

For $\A \in \Mod$ we observe that
\begin{eqnarray}
\label{eq:vertex2}
\vert \{F \subseteq \A : F \mbox{ vertex induced $\A$-face} \} \vert \leq \vert \{x \in \R^\Rxn : x \mbox{ is a vertex of $P$}\} \vert,
\end{eqnarray}
which completes the ingredients that bring us to our crucial theorem.
\begin{theorem} \label{thm:merge_time}
Assume $P$ is bounded.
 Let $\A,\B,\C \in \Mod$ with $\C = \A \dot\cup \B$. Assume the set of injective $\A$-faces, denoted $\mathcal{F}^\A$ and the set of injective $\B$-faces, denoted $\mathcal{F}^\B$, are given. Then the set of injective $\C$-faces, denoted $\mathcal{F}^\C$, can be computed in time
\begin{align*}
 O\left(|\mathcal{V}|^{2k+2} t\right),
\end{align*}
where $t$ is the time needed to check if a subset of $\C$ is an injective $\C$-face.
\begin{proof}
 As argued in the proof of Theorem~\ref{thm:correctness_alg_decomp2} every injective $\C$-face can be obtained from a combination of an injective $\A$-face and an injective  $\B$-face. Hence, we need to consider at most $|\mathcal{F}^\A| \cdot |\mathcal{F}^\B|$ combinations.
Proposition~\ref{claim4} together with (\ref{eq:vertex2}) yields
\begin{align*}
 |\mathcal{F}^\A| \cdot |\mathcal{F}^\B| \leq |\mathcal{V}|^{k+1} \cdot |\mathcal{V}|^{k+1} \leq |\mathcal{V}|^{2k+2}.
\end{align*}
For each combination checking in time $t$ if it defines a face and if it is injective, leads to the bound on the running time.
\end{proof}
\end{theorem}

We still need to argue that $t$ in the above theorem is polynomial time.
\begin{proposition} \label{prop:check_min_feas}
 Given $\A \in \Mod$, it can be checked in input polynomial time if $F\subset \A$ is an injective $\A$-face.
\begin{proof}
Injectivity of $A_F$ is obviously checked in polynomial time. For testing if $F$ defines a face we just solve the following LP:
\begin{align*}
 \max z\; & \\
\mbox{s.t. } A x &= b \\
x_{\A\setminus F} &= 0 \\
x_j &\geq z &\forall j \in F\\
x &\geq 0
\end{align*}
If the LP is unbounded or the optimal value is greater than $0$ then we have found a solution $x^* \in P$ with $x^*_F > 0$ and $x^*_{\A \setminus F} = 0$, which proves that $F$ is a $\A$-face. Reversely, if $F$ is a $\A$-face, there exists a solution of the LP with $z > 0$ and hence, the optimal value of the LP has to be positive.
\end{proof}
\end{proposition}

This brings us to our main theorem.
\begin{theorem} \label{thm:main_result}
Given a branch decomposition of width $k$ of polytope $P$, for constant $k$, Algorithm~\ref{alg:decomposition2} runs in total polynomial time
 $O \left( |\Rxn| |\mathcal{V}|^{2k} t \right)$,
where $t$ is the time needed to check if a subset of $\C$ is an injective $\C$-face.
\begin{proof}
 As mentioned before, the branch decomposition can be represented by a binary tree of $(k-1)$-modules $\Mod$, rooted at $\Rxn$, with leaves the single element modules.
 We observe that the time spend for the leaves (modules $\A \in \Mod$ with $|\A| = 1$) is $O(|\Rxn| t)$, where $t$ is the time needed to check the face property and for injectivity.
 Let $\mathcal{C}$ be the set of interior nodes of the binary tree, then the total time needed for determining the
 injective faces of all modules corresponding to interior nodes is, by Theorem~\ref{thm:merge_time},
\begin{align*}
 O \left( |\mathcal{C}| |\mathcal{V}|^{2k} t \right).
\end{align*}
By Proposition~\ref{prop:check_min_feas} $t$ grows polynomially in input size. Since the number of internal nodes $|\mathcal{C}|$ of a binary tree is bounded by the number of leaves $|\Rxn|$, the result follows.
\end{proof}
\end{theorem}

%% file: conclusionL4.tex
\section{Conclusion}
\label{sec:conclusion}
%\BC I did not touch this, since it depends on decisions we will make about the content of the rest of the paper. It would be good to skype at some time soon. \EC

In the quest for efficient enumeration of vertices of polytopes, we presented a result that is to the best of our knowledge the first substantial step since the hardness proof for general polyhedra by Khachiyan et al. \cite{khachiyan08} in 2008. We translated the property of branch-width for matroids to polyhedra, leading to the notion of $k$-module.
We showed a strong connection between $k$-modules and $(k+1)$-separators in matroid theory. This way we were able to extend the concept of a decomposition into flux modules in M\"uller and Bockmayr~\cite{mueller13} to a decomposition into $k$-modules by using branch-decompositions.

A branch-decomposition of the columns of the coefficient matrix defining the polytope was turned into a hierarchical family of such $k$-modules, which on its turn allowed for a divide and conquer type method to enumerate all vertices in total polynomial time, polynomial in both the size of the input and the output. The running time is
$O \left( |\Rxn| |\mathcal{V}|^{2k} t \right)$. It remains open if vertex enumeration is fixed parameter tractable (FPT) in $k$.
%In a future paper we will report on such a FPT result, but then, next to bounded branch-width, we need bounded subdeterminants of the coefficient matrix as an extra restriction (c.f. Eisenbrandt et al. \cite{Eisenbrandt20??}, Dadush \& H\"anle \cite{Dadush2015}).

If the vertices of polytopes in general can be enumerated in total polynomial time remains an intriguing open problem.

We notice that the algorithm works also for unbounded polyhedra. Just the run time bound does not hold. The boundedness condition is just required for using Caratheodory's Theorem. However, it may very well be that the vertex enumeration problem remains hard under bounded branch-width for general polyhedra. Hence, the complexity of this enumeration problem remains open.
%\BC Here Thomas could add \EC

We repeat here that the restriction to non-constant variables is not crucial. Variables with constant values in the polytope are easy to detect and then they can be projected out. However, the algorithm would also work in the presence of such constant variables, but if not preprocessed, they may hurt the quality of the branch decomposition.
Some of our theorems related to the equivalence between $k+1$-separation and $k$-module allowed $P$ to be just a subset of a convex set. Notice that in this case, depending on the features of $P$, it may not always be easy to detect variables that are constant in $P$.

We notice that, since it is a divide and conquer algorithm, our enumeration algorithm it highly parallelizable. This may turn out to be a crucial aspect for future vertex enumeration algorithms for metabolic network analysis. However, otherwise the result in this paper is mostly of theoretical importance. Small branch-width decompositions are notoriously hard to find, Oum and Seymour~\cite{oum06}, Ma et al.~\cite{ma13}.

Our $k$-modules, presented as a tool to decompose polyhedra given by inequalities, may be a useful approach for other open questions about polytopes. For example, does it allow for bounding the diameter of decomposable polytopes by a polynomial function of dimension and number of constraints of the polytope? A question that is in the center of attention in polyhedral research ever since Hirsch posed his conjecture. It has recently been proven false by Santos \cite{santos12}, but not dramatically false. Finally, we wonder if it allows a strongly polynomial time algorithm for linear programming over polytopes with bounded branch-width. The existence of a strongly polynomial time algorithm for linear programming is a long standing and important open question in operations research and computational complexity.

%% file: appendix.tex
\appendix
\section{Formulation of the polyhedron} \label{appendix:formulation_polyhedron}
Our vertex enumeration results can be extended to polyhedra of the form $Q = \{x \in \R^\Rxn : A x \leq b\}$ for $A \in \R^{\Met \times \Rxn}$. We observe that $Q$ might not be pointed. However, we can efficiently detect if a polyhedron is pointed and the vertex enumeration problem becomes trivial if it is not pointed. Therefore, we assume in the following that $Q$ is pointed. In this case we can transform $Q$ in the following way:
\begin{enumerate}
 \item We introduce slack variables $s$ and obtain the polyhedron $Q' = \{(x,s) \in \R^\Rxn \times \R^\Met : Ax + s = b, s \geq 0\}$. We observe that $\pr_x Q' = Q$, where $\pr_x$ denotes the projection on the variables $x$. %By construction, we have a bijective mapping between the vertices of $Q$ and the vertices of $Q'$.
 \item Let $\X$ be a left null-space matrix of $A$, i.e. $\X z = 0$ if and only if $z \in \langle A \rangle$. It follows that $\pr_s Q' = \{s \in \R^\Met : \X s = \X b, s \geq 0\}$, which we define as $P$. %Clearly, we have an injective mapping of the vertices of $P$ to the vertices of $Q'$.
\end{enumerate}

\begin{proposition}
 There exists a bijective mapping between the vertices (rays) of $Q$ and the vertices (rays) of $P$.
\begin{proof}
Since $Q$ is pointed, it follows that $A$ is injective. Since $\X$ is a null-space matrix of $A$, it follows that there exists for every $s \in P$ an $x$ such that $Ax + s = b$.
Thus, the map $f : Q \to P$ with $f(x) = b - Ax$ is a bijective linear map. The result follows.
\end{proof}
\end{proposition}

Thus, we can apply our results also to polyhedra $Q$, which requires computing a branch-decomposition of the columns of the matrix $\X$ (or equivalently, through matroid duality, of the rows of the matrix $A$) instead of the columns of the matrix $A$.
%\BC I have deleted the rest of this part of the appendix, since it does not add to the results and we are not teaching matroid theory here. I hope you agree? \EC
%However, \ec we actually do not need to compute a branch-decomposition of $\X$ directly as the following theorem shows.
%
%\begin{theorem}
% Let $Q = \{x \in \R^\Rxn : A x \leq b\}$ be a bounded polyhedron with $A \in \R^{\Met \times \Rxn}$. Assume the branch-width of the matroid represented by the rows of $A$ is bounded by $k$. Then we can enumerate the vertices of $Q$ in total polynomial time.
%\begin{proof}
%Let $\matroid$ be the matroid represented by the columns $\Met$ of $\X$. It is a standard matroid theory result that the dual matroid $\matroid^*$ is represented by the rows of $A$ (see Sec.~2.2 in \cite{oxley11}).
%By Cor.~8.1.5 in \cite{oxley11} we also have that $\lambda_\matroid(\A) = \lambda_{\matroid^*}(\A)$ for all $\A \subseteq \Met$. Thus, the branch-width of $\matroid^*$ is equal to the branch-width of $\matroid$. Therefore, a branch-decomposition for the matroid represented by the rows of $A$ is also a branch-decomposition for $P = \{s : \X s = \X b, s \geq 0\}$. The result follows with Theorem~\ref{thm:main_result}.
%\end{proof}
%\end{theorem}
%
%We want to point out that for this result we are interested in the matroid represented by the rows of $A$, while for the rest of the paper we consider the matroid represented by the columns of $A$.

\section{Deciding if a face is vertex induced is NP-hard even with bounded branch-width} \label{appendix:findVertex}
To prove Theorem~\ref{thm:efm_lambda} below, that deciding if a $\A$-face $F$ is vertex induced is NP-hard for polytopes, even under bounded branch-width, we prove it first for unbounded polyhedra. Then we use the fact that any vertex of $P$ is of polynomial size in the outer description of $P$ (see Schrijver~\cite{Schrijver86}) to bound the polyhedron to a polytope.

For $k$-modules $\A$ that only contain one element the problem of deciding if $\A$ is a vertex induced $\A$-face is equivalent to deciding if a vertex exists with a given variable (element of $\A$) in its support. We refer to this problem as {\sc Find Vertex with Single Support}. We notice that it cannot be solved by linear programming, because the polyhedron may be unbounded.
%It is a generalization of the graph-theoretic problem of finding a circuit (simple cycle) through two given arcs in a directed graph. In this case, we just need to choose $A$ as the incidence matrix of the directed graph and move one of the arcs into the right hand side vector $b$ of the polyhedron.

{\sc Find Vertex with Single Support} (FVSS): \\
{\em Instance}: Polyhedron $P = \{x \in \R^\Rxn : Ax = b, x \geq 0\}$ and one variable $e \in \Rxn$. \\
{\em Question}: Does there exist a vertex $x$ of $P$ with $x_e > 0$?

It was shown in Acu\~na et al.~\cite{Acuna2010210} that this problem is NP-hard in general. Here we show that FVSS remains hard in case of bounded branch-width.

\begin{theorem} \label{thm:efm_lambda}
 Problem FVSS is \cNP-hard for polyhedra $P$ of branch-width $3$.
\end{theorem}
Our proof is by reduction from SCIPF, which was already used in Cunningham and Geelen~\cite{cunningham07} to show hardness for general IP with bounded branch-width.
% \BC Was this problem also used to prove hardness by Cunningham and Geelen. Then we may refer to it \EC

{\sc Single Constraint Integer Programming Feasibility} (SCIPF): \\
{\em Instance}: A non-negative vector $a \in \mathbb{Z}^n$ and an integer $b$. \\
{\em Question:} Does there exist $x \in \mathbb{Z}^n$ satisfying $ax = b, x \in \{0,1\}^n$?

%\BCMP maybe we should add a reference to this, even though it's a classical NP-hard problem ... e.g., garey and johnson? \ECMP

Given an instance of SCIPF we create a polyhedron $P$ with a designated variable $x_0$ such that $P$ contains a vertex with $x_0 > 0$ if and only if there exists an $x \in \{0,1\}^n$ with $a x = b$.
Consider the polyhedron $P$ defined by the feasible solutions $(x,x_0,y,z)$ of the following inequalities:
\begin{align}
 ax + x_0 - (b+1) z &= 0 \label{eq:P1}\\
\Eins y + \Eins x - nz &= 0  \label{eq:P2}\\
x_i - z &\leq 0 &\mbox{for } i \in \{1, \ldots, n\} \label{eq:P3}\\
y_i - z &\leq 0 &\mbox{for } i \in \{1, \ldots, n\} \label{eq:P4}\\
2z - x_0 &= 1  \label{eq:P5}\\
x_i &\geq 0 &\mbox{for } i \in \{1, \ldots, n\} \label{eq:P6}\\
y_i &\geq 0 &\mbox{for } i \in \{1, \ldots, n\} \label{eq:P7}\\
x_0 &\geq 0  \label{eq:P8}\\
z &\geq 0 \label{eq:P9},
\end{align}
where $\Eins$ denotes the all-1 vector.

We will show that solutions $x \in \{0,1\}^n$ with
$ax = b$ correspond to vertices of $P$ with $x_0 > 0$ (Theorem~\ref{thm:efm_reduction}) and that $P$ has bounded branch-width  (Theorem~\ref{thm:efm_branch_width}), which allows us to conclude that FVSS is \cNP-hard on polyhedra of bounded branch-width (Theorem~\ref{thm:efm_lambda}).

\begin{lemma} \label{lemma:x_boolean}
 Let $(x,x_0,y,z)$ be a vertex of $P$ with $x_0 > 0$.
%W.l.o.g. there exists a $k$ with $y_i = 0$ for all $i \leq k$ and $y_i > 0$ for all $i > k$.
 %$\sum_{i: x_i > 0} a_i < b+1$.
It then holds for all $i \in \{1, \ldots, n\}$ that $x_i \in \{0,z\}$.
\begin{proof}
% If there exists no $i$ with $0 < x_i < z$, the claim follows because $x_0 > 0$.
% Hence assume there exists a $i$ with $0 < x_i < z$.
Assume the lemma is false. Then there exists a smallest counterexample, i.e., a vertex with the largest number of inequality constraints being satisfied with equality. Choose $j \in \arg\max\{a_i : 0 < x_i < z\}$. We distinguish two cases.
\begin{itemize}
 \item[Case $1$:] \emph{There exists a $k \neq j$ with $0 < x_k < z$.} Define $y' = y$ and $x' \in \R^n$ by
\begin{align*}
 x'_i &= x_i &\mbox{ for } k \neq i \neq j \\
 x'_j &= x_j + \min(z-x_j, x_k) \\
 x'_k &= x_k - \min(z-x_j, x_k)
\end{align*}
Because $a_j \geq a_k$, it follows that
\begin{align}
 ax' + x_0 &= ax + (a_j - a_k) \min(z-x_j, x_k) + x_0 \geq ax + x_0 = (b+1)z  \label{eq:x_boolean_case1}\\
 \Eins y + \Eins x' &= nz. \notag
\end{align}
Hence $(x', x_0, y', z)$ satisfies all constraints of $P$ except \eqref{eq:P1}. Moreover, (besides \eqref{eq:P1}) a strictly larger set of inequality constraints is satisfied by equality: $x'_j=z$ or $x'_k=0$.
\item[Case $2$:] \emph{There exists no $k \neq j$ with $0 < x_k < z$.} Thus $\Eins x$ is not a multiple of $z$ and therefore, also $\Eins y$ is not a multiple of $z$. It follows that there exists a $k$ with $0 < y_k < z$. Define $x', y' \in \R^n$ by:
\begin{align*}
 x'_i &= x_i &\mbox{ for } i \neq j \\
 x'_j &= x_j + \min(z-x_j, y_k) \\
 y'_i &= y_i &\mbox{ for } i \neq k \\
 y'_k &= y_k - \min(z-x_j, y_k)
\end{align*}
Because $a_j \geq 0$, it follows that
\begin{align}
 ax' + x_0 &= ax + a_j \min(z-x_j, y_k) + x_0 \geq ax + x_0 = (b+1)z \label{eq:x_boolean_case2}\\
 \Eins y + \Eins x' &= nz. \notag
\end{align}
Hence $(x', x_0, y', z)$ satisfies all constraints of $P$ except \eqref{eq:P1}. Also in this case (besides \eqref{eq:P1}) a strictly larger set of inequality constraints is satisfied by equality.
\end{itemize}
In both cases, we found $(x', x_0, y', z)$ that satisfies all constraints of $P$ except \eqref{eq:P1}. Additionally, a strictly larger set of inequality constraints is satisfied by equality (excluding \eqref{eq:P1}) than by $(x, x_0, y, z)$. Hence, if also \eqref{eq:P1} were satisfied with equality by $(x', x_0, y', z)$ we would have a counterexample to $(x, x_0, y, z)$ being a vertex with largest number of inequality constraints being met with equality and the lemma is proved.

Therefore, assume $ax' + x_0 > (b+1)z$. For $\theta \geq 0$ we now define:
\begin{align*}
 x''(\theta) &:= \theta x' \\
 y''(\theta) &:= \theta y' \\
 x''_0(\theta) &:= 2 \theta z - 1 \\ %= 2 (\theta \frac{1}{2}(x_0 + 1) - 1) = \theta x_0 + \frac{1}{2} (1 - \theta)\\
 z''(\theta) &:= \theta z.
\end{align*}
Observe that $x''_0(1) = x_0$.

%\BC Does this really have to go through this case analysis? \EC
Also observe that all constraints that do not involve $x_0$ are homogenous and hence are also satisfied by $(x''(\theta),\allowbreak x''_0(\theta),\allowbreak y''(\theta), z''(\theta))$ for all $\theta \geq 0$. Moreover, if an inequality constraint was satisfied by equality with $(x',x_0,y',z)$ it will also be satisfied by equality with $(x''(\theta), x''_0(\theta), y''(\theta), z''(\theta))$. Constraint \eqref{eq:P5} is also satisfied by construction. Because of constraint \eqref{eq:P8} we are allowed to choose $\theta$ only from the range $\theta \geq \frac{1}{2z}$.

Suppose that there exists a $\theta\geq \frac{1}{2z}$ such that $a x''(\theta) + x''_0(\theta) - (b+1) z''(\theta) = 0$. Then $(x''(\theta), x''_0(\theta), y''(\theta), z''(\theta))$ is a feasible solution and satisfies more inequality constraints by equality than $(x,x_0, y, z)$, a contradiction.

Thus we only need to consider the case where $a x''(\theta) + x''_0(\theta) - (b+1) z''(\theta) \neq 0$ for all $\theta \geq \frac{1}{2z}$. Since $a x''(\theta) + x''_0(\theta) - (b+1) z''(\theta)$ is continuous in $\theta$ and $a x''(1) + x''_0(1) - (b+1) z''(1) \geq 0$ by \eqref{eq:x_boolean_case1} and \eqref{eq:x_boolean_case2} it follows that $a x''(\theta) + x''_0(\theta) - (b+1) z''(\theta) > 0$ for all $\theta \geq \frac{1}{2z}$.
In particular, for $\theta= \frac{1}{2z}$ we have $a x''(\frac{1}{2z}) + x''_0(\frac{1}{2z}) - (b+1) z''(\frac{1}{2z}) > 0$. Notice that $x''_0(\frac{1}{2z})=0$, whence $a x''(\frac{1}{2z}) - (b+1) z''(\frac{1}{2z}) > 0$.
Since $a x + x_0 - (b+1) z = 0$ and $x_0 > 0$ we have $a x - (b+1) z < 0$.
Thus, there exists a $\nu \in [0,1]$ such that for
\begin{align*}
x''' &:= \nu x + (1-\nu) x''(\frac{1}{2z}) \\
%x_0''' &:= \nu x_0 + (1-\nu) x_0''(\frac{1}{2z}) \\
x_0''' &:= 0 \\
y''' &:= \nu y + (1-\nu) y''(\frac{1}{2z}) \\
z''' &:= \nu z + (1-\nu) z''(\frac{1}{2z})
\end{align*}
we have
$$ a x''' + x_0''' - (b+1) z''' = 0 $$

It follows that $(x''', x_0''', y''', z''')$ satisfies all constraints of $P$ except constraint \eqref{eq:P5}.  However, since $z''' \geq \frac{1}{2}$, using the homogeneity arguments as before, it follows for $\sigma = \frac{1}{2z'''}$ that $(\sigma x''', \sigma x_0''', \sigma y''', \sigma z''') \in P$. We observe that $(\sigma x''', \sigma x_0''', \sigma y''', \sigma z''')$ satisfies more inequality constraints by equality than $(x,x_0,y,z)$, again a contradiction.
\end{proof}
\end{lemma}

\begin{theorem} \label{thm:efm_reduction}
 $P$ has a vertex with $x_0 > 0$ if and only if there exists a solution of $a x = b$ with $x \in \{0,1\}^n$.
\begin{proof}
 \begin{itemize}
  \item[$\Leftarrow$:] Let $x \in \{0,1\}^n$ with $a x = b$. Define:
\begin{align*}
 y_i &= \begin{cases}
         0 & \mbox{if }i \leq \Eins x \\
				 1 & \mbox{if }i > \Eins x
        \end{cases} &\mbox{for } i \in \{1, \ldots, n\}\\
 x_0 &= 1 \\
 z &= 1
\end{align*}
Note that $\Eins x$ just counts the number of non-zero entries in $x$.

Clearly the inequality constraints are satisfied and we also have:
\begin{align*}
 ax + x_0 - (b+1) z &= b + 1 - (b+1) \cdot 1 = 0 \quad \eqref{eq:P1}\\
\Eins y + \Eins x - nz &= n - \Eins x + \Eins x - n = 0 \quad \eqref{eq:P2} \\
2z -x_0 &= 2 - 1 = 1 \quad \eqref{eq:P5}
\end{align*}
Thus $x,x_0,y,z \in P$.
A little though should make it clear that amongst the set of constraints that are tight in $(x,x_0,y,z)$ $2n+2$ of them are linearly independent, whence $(x,x_0,y,z)$ is a vertex of $P$.
%Furthermore, it is easy to see that the solution has minimal support:
%\begin{itemize}
%\item $z$ cannot be dropped by a combination of constraints \eqref{eq:P5} and \eqref{eq:P8}.
%\item $x_0 = 0$ would imply that
%\begin{align*}
% ax - (b+1)z & \leq az - (b+1) z = -z < 0,
%\end{align*}
%contradiction \eqref{eq:P1}.
%\item Reducing the support of $x$ or $y$, would give:
%\begin{align*}
% \Eins x + \Eins y \leq (n-1)z,
%\end{align*}
%contradicting \eqref{eq:P2}.
%\end{itemize}
%Since $x_i - z = 0$ for all $i$ with $x_i > 0$ and $y_i - z = 0$ for all $i$ with $y_i > 0$, it follows that additional inequality constraints in \eqref{eq:P3} and \eqref{eq:P4} can only be satisfied with equality by increasing the support of $x$ or $y$.
%Hence $(x,x_0,y,z)$ is a vertex of $P$.

\item[$\Rightarrow$:] Let $(x,x_0,y,z)$ be a vertex of $P$ with $x_0 > 0$.
By Lemma~\ref{lemma:x_boolean} we know that $x_i \in \{0, z\}$ for all $i$. Since $x_0 = 2z -1$ we have
\begin{align*}
 0 < \frac{x_0}{z} < 2.
\end{align*}
Hence it follows from $ax + x_0 = (b+1)z$ that:
\begin{align*}
 \sum_{i : x_i > 0} a_i &= \frac{ax}{z} < \frac{ax}{z} + \frac{x_0}{z} = b+1 \\
 \sum_{i : x_i > 0} a_i &= \frac{ax}{z} > \frac{ax}{z} + \frac{x_0}{z}-2 = b-1.
\end{align*}
Since $a,b$ are integer, it follows that:
\begin{align*}
 \sum_{i : x_i > 0} a_i &= b.
\end{align*}
We thus have a solution of $ax = b$ with $x \in \{0,1\}^n$.
 \end{itemize}
\end{proof}
\end{theorem}

As a last step, it remains to show that the coefficient matrix of $P$ has bounded branch-width.
We therefore assume that for the inequalities \eqref{eq:P3} and \eqref{eq:P4}, we introduce slack-variables $x^s, y^s$ to also turn them into equality constraints.

\begin{theorem} \label{thm:efm_branch_width}
 The coefficient matrix of $P$ has a branch-width of at most $3$.
\begin{proof}
 Let $A$ be the coefficient matrix of $P$ without the column corresponding to the variable $z$. To build the branch-decomposition for $A$, we first build the subsets $X_i := \{x_i, x_i^s\} \in \Mod$ and $Y_i := \{y_i, y_i^s\} \in \Mod$ for all $i \in \{1, \ldots, n\}$.
We observe that for each $i \in \{1, \ldots, n\}$ the variable $x_i^s$ (resp. $y_i^s$) only appears in row \eqref{eq:P3} (resp. \eqref{eq:P4}) and in this row the only other variable is $x_i$ (resp. $y_i$), because we ignore the column for the variable $z$.
%We observe that $x_i$ and $x_i^s$ (resp. $y_i$ and $y_i^s$) are coparallel in $M(A)$ ($M(A)$ is the matroid represented by $A$).
Hence, it suffices to track in which constraints the variables $x,y$ appear. This is only the case in rows \eqref{eq:P1} and \eqref{eq:P2}. The variable $x_0$ appears alone in row \eqref{eq:P5} and hence, this row can be ignored as well. %is a coloop of $M(A)$.

Hence, we can combine the subsets $X_i, Y_i$ for $i \in \{1, \ldots, n\}$ and $x_0$ in arbitrary order to form the family $\Mod$ of subsets. It is easy to see that any such combination of subsets of columns has
connectivity function value at most $2$. Hence, we obtain a branch-decomposition of branch-width at most $2$ for the coefficient matrix $A$ without the $z$-column.

To obtain a branch-decomposition for the whole coefficient matrix of $P$, we add the column of $z$ at an arbitrary place in the branch-decomposition. Thus, we obtain a branch-decomposition of branch-width at most $3$.
\end{proof}
\end{theorem}

\begin{proof}[Theorem~\ref{thm:efm_lambda}]
 We can reduce SCIPF to finding a vertex $x$ of a polyhedron $P$ with $x_i > 0$ for a given variable $x_i$ by Theorem~\ref{thm:efm_reduction}. By Theorem~\ref{thm:efm_branch_width}, we know that $P$ has branch-width at most $3$.
 Thus, we can solve SCIPF by solving FVSS on a polyhedron of branch-width at most $3$.
\end{proof}

Using results from Schrijver~\cite{Schrijver86} (Chapter 10) we know that all vertices of $P$ in the proof above have a size bounded by the size of the input to describe $P$ in its outer form (given above). Therefore, we can construct a bounded polyhedron $Q := \{(x,y,z,w) : z + w = \Gamma, (x,y,z) \in P, w \geq 0\}$, where $\Gamma$ is sufficiently large (but with binary encoding size polynomially bounded by the input size) such that $z < \Gamma$ for every vertex $(x,y,z)$ of $P$.

We observe that $Q$ is bounded and $(x, y, z, w)$ is a vertex of $Q$ with $x_0 > 0$ and $w > 0$ (i.e. $F = \{x_0, w\}$ is a vertex induced $F$-face of $Q$) if and only if $(x,y,z)$ is a vertex of $P$ with $x_0 > 0$. We observe that for this construction the branch-width does not go up and we get:

\begin{theorem} \label{thm:hardness_polytopes}
 Given a polytope $Q$ and a $K$-face $F$ it is \cNP-hard to decide if $F$ is vertex induced even if $Q$ has branch-width at most $3$ and $|K| = 2$.
\end{theorem}